\documentclass[11pt]{article}
\usepackage[margin=1in]{geometry} 
\usepackage{amsmath}
\usepackage{amssymb,amsfonts,bbm,mathrsfs,stmaryrd}
\usepackage{url}

\usepackage[amsmath,thmmarks]{ntheorem}
\usepackage{hyperref}
\usepackage{cleveref}

\theoremstyle{change}

\newtheorem{definition}[equation]{Definition}

\newtheorem{theorem}[equation]{Theorem}
\newtheorem{prop}[equation]{Proposition}
\newtheorem{proposition}[equation]{Proposition}
\newtheorem{lemma}[equation]{Lemma}
\newtheorem{cor}[equation]{Corollary}

\newtheorem{example}[equation]{Example}

\theorembodyfont{\upshape}
\theoremsymbol{\ensuremath{\Diamond}}

\newtheorem{remark}[equation]{Remark}

\theoremstyle{nonumberplain}

\theoremsymbol{\ensuremath{\Box}}
\newtheorem{proof}{Proof}

\qedsymbol{\ensuremath{\Box}}

\creflabelformat{equation}{#2(#1)#3} 

\crefname{equation}{equation}{equations}
\crefname{eg}{example}{examples}
\crefname{defn}{definition}{definitions}
\crefname{prop}{proposition}{propositions}
\crefname{thm}{Theorem}{Theorems}
\crefname{lemma}{lemma}{lemmas}
\crefname{cor}{corollary}{corollaries}
\crefname{remark}{remark}{remarks}
\crefname{section}{Section}{Sections}
\crefname{subsection}{Section}{Sections}

\crefformat{equation}{#2equation~(#1)#3} 
\crefformat{eg}{#2example~#1#3} 
\crefformat{defn}{#2definition~#1#3} 
\crefformat{prop}{#2proposition~#1#3} 
\crefformat{thm}{#2Theorem~#1#3} 
\crefformat{lemma}{#2lemma~#1#3} 
\crefformat{cor}{#2corollary~#1#3} 
\crefformat{remark}{#2remark~#1#3} 
\crefformat{section}{#2Section~#1#3} 
\crefformat{subsection}{#2Section~#1#3} 

\Crefformat{equation}{#2Equation~(#1)#3} 
\Crefformat{eg}{#2Example~#1#3} 
\Crefformat{defn}{#2Definition~#1#3} 
\Crefformat{prop}{#2Proposition~#1#3} 
\Crefformat{thm}{#2Theorem~#1#3} 
\Crefformat{lemma}{#2Lemma~#1#3} 
\Crefformat{cor}{#2Corollary~#1#3} 
\Crefformat{remark}{#2Remark~#1#3} 
\Crefformat{section}{#2Section~#1#3} 
\Crefformat{subsection}{#2Section~#1#3}

\numberwithin{equation}{section}

\usepackage{listings}
\usepackage{caption}
\DeclareCaptionFont{white}{\color{white}}
\DeclareCaptionFormat{listing}{%
  \parbox{\textwidth}{\colorbox{gray}{\parbox{\textwidth}{#1#2#3}}\vskip-4pt}}
\usepackage{tikz}
\tikzset{dot/.style={circle,draw,fill,inner sep=1pt}}
\usepackage{braids}
\usepackage{tqft}
\usetikzlibrary{tqft}
\usetikzlibrary{cd}
\usetikzlibrary{arrows}

\newcommand\ket[1]{\mid #1 \rangle}
\newcommand\bra[1]{\langle #1 \mid}
\newcommand\setof[1]{\{ #1 \}}

\newcommand\abs[1]{ \mid #1 \mid }

\title{Some Coxeter Groups in Reversible and Quantum Computation}
\author{Jon Aytac, Ammar Husain}

\begin{document}
\maketitle

\abstract

In this article we show how the structure of Coxeter groups are present in gate sets of reversible and quantum computing. These groups have efficient word problems which means that circuits built from these gates have potential to be shortened efficiently. This is especially useful in the case of quantum computing when one does not have the timescale to perform a long series of gates and so one must find a gate scheduling that minimizes circuit depth. As the main example we consider the oracle for 3SAT.

\section{Introduction}

\subsection{Coxeter Groups Background}

In this section let us set up the general notation we will use for the Coxeter groups we will use in both reversible and quantum gates. 

\begin{definition}[Coxeter Group]
A Coxeter group $(W,R)$ is a group $W$ equipped with a presentation by involutions $r_i$ as:

\begin{eqnarray*}
G &=& \langle r_1 \cdots r_n \mid (r_i r_j)^{m_{ij}} = 1 \rangle
\end{eqnarray*}

For $i \neq j$ $m_{ij} \in \mathbb{N}_{\geq 2} \bigsqcup \infty$. $m_{ij}=\infty$ means no relation should be imposed for the order of $r_i r_j$.

This data is encoded in a graph where the vertices are the generators $r_i$. If $m_{ij}>2$, there is an edge between $r_i$ and $r_j$ colored by $m_{ij}$. If $m_{ij}=2$, there is no edge.

\end{definition}

\begin{example}[$S_n$]
The symmetric group with generators $r_i = (i,i+1)$ forms a Coxeter group with Coxeter graph with $n-1$ vertices connected in a line.
\end{example}

\begin{definition}[Crystallographic]
A Coxeter group $(W,R)$ is called crystallographic if there exists a lattice $L$ such that $g (L) \subset L$ for all $g \in W$. This gives the group a representation over the integers.
\end{definition}

\begin{lemma}[Stabilizer Code]
Given a Coxeter graph, finding abelian subgroups $A$ is solved finding independent sets. A stabilizer code is then the intersection of the $+1$ eigenspaces of all these generators.
\end{lemma}

\begin{remark}
A good error correcting code constructed this way starts with finding a good set of generators for your system. They should be the involutions that you expect to happen by mistake in your system (for example, the single (qu)bit flips). The choice of generators below is not a good one because it pushes all the interesting behavior onto the first few indices. However, the Coxeter graph for another choice of generators can easily be built from this data.
\end{remark}

One can do many things with arbitrary Coxeter groups like construct Hecke algebras, Soergel bimodules and more conjecturally categories of unpotent character sheaves \cite{Geordie}. The class of examples we list below are no different and one can do the same for potentially interesting algebras and Kazhdan-Lusztig polynomials. But here we will focus on the rewriting problem. Individual computational gates will be generators of a Coxeter group and as such we define the following:

\begin{definition}[Coxeter Compiler]
An algorithm that takes a word in a Coxeter group representing a computation and outputs the reduced normal form. This shortens the programs from long unoptimized presentations into a circuit that implements the same operation. For example, one may write  a reversible computation as $(x,y) \to (x,y+x)$ followed by $(z,w) \to (z,z+w)$ without realizing that it simplifies.
\end{definition}

\subsection{2-Category of Programs}

Computation is naturally viewed in a 2-categorical language. This section is not necessary for the main results, but we hope it shows how unpacking a higher categorical perspective helps in very concrete rigid problems. It quite literally gives an extra dimension for imposing compositionality.

\begin{remark}
This conflicts with the philosophy of homoiconicity, but that does not hold in general when one does not have exponential objects.
\end{remark}

\begin{definition}[Proceess of Computation \cite{Stay}]
Data types are objects, programs are morphisms and equivalence classes of reductions are 2-morphisms.
\end{definition}

In our cases, the only data types are indexed by natural numbers and correspond to either $n$ bits or $n$ qubits. The programs are words in the relevant Coxeter group and the 2-morphisms are simplifications done by the Coxeter compiler.

Even upon deformation to Hecke algebras, this fits nicely with the perspective of 2 dimensional field theories. A 2-dimensional topological field theory is defined by $Alg_{bim}$, the $(\infty,2)$ category with algebras, bimodules and intertwiners. For our case, the words in a Coxeter group turn into Soergel bimodules, composition becomes tensor product of bimodules and the reductions become isomorphisms thereof. Note here we have only kept the invertible 2-morphisms even though others were allowed a priori.

Remembering that $n$ and $m$ can be combined into $n+m$ allows one to bring in a small third dimension. A program that operates on the $n$ and $m$ bits separately and then combines them can be thought of in terms of the 2-dimensional topological field theory for $n$, $m$ and then merge the sheets into the one for $n+m$. This gives a discretized third dimension. One may think of each sheet being different agents who are forking and merging.

\begin{definition}[Phyllo field theory]
This name is nonstandard but the picture to have in mind is old \cite{KashaevKolya}. The type $A$ example illustrates this most clearly. For each particular $A_n$ we can draw pictures on a single layer with Soergel bimodules and their morphisms. But the inclusion $S_n \times S_m \to S_{n+m}$ allow one to combine two layers into one. This can be repeated as much as desired in order to make a flakey structure.
\end{definition}

\begin{remark}[Error correction functor]
Let $C$ be an error correcting code that replaces $k$ (qu)bits by $n$ (qu)bits. Let it's encoding and decoding circuits be given in the same gate set $G$. Consider the category $\mathcal{C}_{k,G}$ whose objects are indexed by multiples of $k$ for possible number of (qu)bits and whose morphisms are potentially unreduced words in the gate set $G$. Then applying the error correcting code gives a functor to $\mathcal{C}_{n,G}$. The objects $k*l \to n*l$ and each gate in the word gets sandwiched between a decoding and an encoding. In a functorial field theory picture this means that the error correcting code gives a manifold of one dimension higher with two boundaries for the $\mathcal{C}_{k,G}$ and $\mathcal{C}_{n,G}$ parts respectively.
\end{remark}

\section{Classical Gates}

\subsection{Universal Gates}

Let us set up the notation for the types of universal gates used in reversible computing.

\begin{definition}[Toffoli Gate]
A 3 bit gate which if the first 2 bits are set to true, then the third bit is flipped. Otherwise nothing gets changed. Index these as $Tof^n_{ijk}$ where $ijk$ index which $3$ bits lines of the total $n$ are used. There is symmetry under swapping $ij$.
\end{definition}

\begin{definition}[Fredkin Gate]
A 3 bit gate that implements a controlled swap. If the first bit is set to true then the second and third are swapped. Otherwise nothing gets changed. Index these as $Fred^n_{ijk}$. There is symmetry under swapping $jk$.
\end{definition}

\begin{definition}[$C^{k,n}$]
A $k$ bit gate which maps $0^k$ to $1^k$ and vice versa. All other $k$ bit inputs are sent to themselves. The generators are $C^{k,n}_{i_1 \cdots i_k}$ where $i_1 \cdots i_k$ have full permutation symmetry of $S_k$.
\end{definition}

\begin{definition}[$T^{k,n}$]
A $k$ bit gate which maps $x$ to $\bar{x}$ if the Hamming weight $\abs{x}$ is odd. All other inputs are sent to themselves. The generators are $T^{k,n}_{i_1 \cdots i_k}$ where $i_1 \cdots i_k$ have full permutation symmetry of $S_k$.
\end{definition}

\begin{definition}[$F^{k,n}$]
A $k$ bit gate which maps $x$ to $\bar{x}$ if the Hamming weight $\abs{x}$ is even. All other inputs are sent to themselves. The generators are $F^{k,n}_{i_1 \cdots i_k}$ where $i_1 \cdots i_k$ have full permutation symmetry of $S_k$.
\end{definition}

\begin{definition}[$CNOTNOT^n$]
A 3 bit gate that implements a controlled not. If the first bit is set to true, then a not is applied to both the second and third. Otherwise nothing gets changed. Index these as $CNOTNOT^n_{ijk}$. There is symmetry under swapping $jk$.
\end{definition}

\subsection{Universality Result}

In this portion, the relevant definitions and main result of \cite{Aaronson} are reviewed.

\begin{definition}[Ancilla]
For universality results on $n$ bits, you are allowed to put in as $O(1)$ ancilla bits as you want provide you return their values to their initializations at the end. The bound for this constant depends on the gate set.
\end{definition}

\begin{definition}[Computational Generated Group]
Given any set of reversible gates $\setof{G_i}$, let $G^n$ be the smallest subgroup of $S_{2^n}$ containing

\begin{itemize}
\setlength\itemsep{-1em}
\item All the $G_i$\\
\item All permutations of $n$\\
\item The inclusions of $G^{n-1}$ upon adjoining a dummy bit line.\\
\end{itemize}

and satisfying the ancilla rule that whenever $G^{n+k}$ contains a transformation $F$ that is the identity on the last $k$ components and does not depend on them as $F(x_1 \cdots x_n a_1 \cdots a_k) = H(x_1 \cdots x_n) a_1 \cdots a_k$, then $H$ is an element of $G^n$.

\end{definition}

\begin{theorem}

All generated groups are one of the following:

\begin{itemize}
\setlength\itemsep{-1em}
\item Only bit swaps. - computationally generated by the empty set\\
\item All transformations - computationally generated by $Tof^{n}_{ijk}$\\
\item Hamming weight conserving subgroup - computationally  generated by Fredkin $Fred^n_{ijk}$\\
\item The class of all modulo $k$ preserving transformations-computationally generated by $C^{k,n}_{i_1 \cdots i_k}$\\
\item All affine transformations - computationally generated by $CNOT^n_{ij}$\\
\item All parity preserving affine transformations - computationally generated by $CNOTNOT^n_{ijk}$\\
\item All mod-4 preseving affine transformations - computationally generated by $F^{4,n}_{i_1 \cdots i_4}$\\
\item All orthogonal linear transformations - computationally generated by $T^{4,n}_{i_1 \cdots i_4}$\\
\item All mod-4 preserving orthogonal linear transformations - computationally generated by $T^{6,n}_{i_1 \cdots i_6}$\\
\item NOTNOT Augmented class 1 - computationally generated by the $??$ and $NOTNOT_{ij}$.\\
\item NOTNOT Augmented class 3 - computationally generated by the $Fred^{n}_{ijk}$ and $NOTNOT_{ij}$.\\
\item NOTNOT Augmented classes 7/8 - computationally generated by either $F^{4,n}_{i_1 \cdots i_4}$ and $NOTNOT_{ij}$ or $T^{4,n}_{i_1 \cdots i_4}$ and $NOTNOT_{ij}$. The two choices are equivalent groups.\\
\item NOTNOT Augmented classes 9 - computationally generated by either $T^{6,n}_{i_1 \cdots i_6}$ and $NOTNOT_{ij}$\\
\item NOT Augmented class 1 - computationally generated by the $??$ and $NOT_{i}$.\\
\item NOT Augmented class 3 - computationally generated by the $Fred^{n}_{ijk}$ and $NOT_{i}$.\\
\item NOT Augmented class 6 - computationally generated by the $CNOTNOT^{n}_{ijk}$ and $NOT_{i}$.\\
\item NOT Augmented classes 7/8 - computationally generated by either $F^{4,n}_{i_1 \cdots i_4}$ and $NOT_{i}$ or $T^{4,n}_{i_1 \cdots i_4}$ and $NOT_{i}$. The two choices are equivalent groups.\\
\item NOT Augmented classes 9 - computationally generated by $T^{6,n}_{i_1 \cdots i_6}$ and $NOT_{i}$\\
\end{itemize}
\end{theorem}

\begin{remark}
The class of all reversible computations computationally generated by $Tof_{ijk}$ can be understood categorically\cite{CockettComfort}.
\end{remark}

\subsection{Coxeter Structure}

\begin{cor}
Each class comes with the structure of quotients from a countable family of Coxeter groups.
\end{cor}

\begin{proof}
Everything in $G^n$ can be regarded as a circuit with some number $k$ of ancilla using only the prescribed gates $G_i$ and the permutations of $n+k$. Because there are a finite number of reversible transformation of $n$ bits, for each $n$ there is a minimal number $k$ so WLOG take $k$ to be this minimum necessary for all of $G^n$. The $G_i$ is taken to be acting on only the first few bits because conjugation can provide the others.\\
All of these generators are involutions. The product of any two is an element of a finite group $S_{2^{n+k}}$. Computing these orders gives the requisite Coxeter matrix. There may be other relations, but the Coxeter group has a map to that quotient by the additional relations.\\
In fact often the Coxeter group will be infinite as witnessed by affine Coxeter subgraphs. That guarantees a drastic quotient to get down to $\leq (2^{n+k})!$.
\end{proof}

For each of these classes, we should provide the sequence of Coxeter groups indexed by $n$ which we assume is $\geq 4$ to avoid the degenerate cases.

The first case with only bit swaps is simply $S_n$ with the usual Coxeter generators $s_{i,i+1}$.\\

We can do all these cases except the countable family $C^{k,n}$ at once.

\begin{theorem}
The Coxeter structure from the generators on 

\begin{itemize}
\setlength\itemsep{-1em}
\item $r_1=s_{12}$\\
\item $r_2=s_{23}$\\
\item $r_{i \leq n}=s_{i,i+1}$\\
\item $r_{n}=NOT_1$\\
\item $r_{n+1}=CNOT_{12}$\\
\item $r_{n+2}=Tof_{123}$\\
\item $r_{n+3}=Fred_{123}$\\
\item $r_{n+4}=CNOTNOT_{123}$\\
\item $r_{n+5}=F^{4,n}_{1234}$\\
\item $r_{n+6}=T^{4,n}_{1234}$\\
\item $r_{n+7}=T^{6,n}_{123456}$\\
\item $r_{n+8}=NOTNOT_{12}=NOT_1 NOT_2$\\
\end{itemize}

with $n=7$ is

\begin{eqnarray*}
M &=& \left(
\begin{array}{ccccccccccccccc}
 1 & 3 & 2 & 2 & 2 & 2 & 4 & 3 & 2 & 6 & 3 & 2 & 2 & 2 & 2 \\
 3 & 1 & 3 & 2 & 2 & 2 & 2 & 4 & 6 & 2 & 2 & 2 & 2 & 2 & 4 \\
 2 & 3 & 1 & 3 & 2 & 2 & 2 & 2 & 4 & 6 & 4 & 2 & 2 & 2 & 2 \\
 2 & 2 & 3 & 1 & 3 & 2 & 2 & 2 & 2 & 2 & 2 & 3 & 3 & 2 & 2 \\
 2 & 2 & 2 & 3 & 1 & 3 & 2 & 2 & 2 & 2 & 2 & 2 & 2 & 2 & 2 \\
 2 & 2 & 2 & 2 & 3 & 1 & 2 & 2 & 2 & 2 & 2 & 2 & 2 & 3 & 2 \\
 4 & 2 & 2 & 2 & 2 & 2 & 1 & 4 & 4 & 4 & 4 & 4 & 4 & 4 & 2 \\
 3 & 4 & 2 & 2 & 2 & 2 & 4 & 1 & 4 & 4 & 2 & 3 & 3 & 3 & 4 \\
 2 & 6 & 4 & 2 & 2 & 2 & 4 & 4 & 1 & 3 & 4 & 6 & 6 & 6 & 4 \\
 6 & 2 & 6 & 2 & 2 & 2 & 4 & 4 & 3 & 1 & 2 & 4 & 4 & 4 & 4 \\
 3 & 2 & 4 & 2 & 2 & 2 & 4 & 2 & 4 & 2 & 1 & 4 & 4 & 4 & 4 \\
 2 & 2 & 2 & 3 & 2 & 2 & 4 & 3 & 6 & 4 & 4 & 1 & 2 & 2 & 2 \\
 2 & 2 & 2 & 3 & 2 & 2 & 4 & 3 & 6 & 4 & 4 & 2 & 1 & 2 & 2 \\
 2 & 2 & 2 & 2 & 2 & 3 & 4 & 3 & 6 & 4 & 4 & 2 & 2 & 1 & 2 \\
 2 & 4 & 2 & 2 & 2 & 2 & 2 & 4 & 4 & 4 & 4 & 2 & 2 & 2 & 1 \\
\end{array}
\right)
\end{eqnarray*}

For more than $7$ bits, the new generators are $r_i$ for $i=7$ through $n-1$. These commute with all the $r_n \cdots r_{n+8}$ because they act on different bits. Thus $m_{ij}=2$ for these. The only nontrivial $m_{ij}$ are with those generators that correspond to $r_{i-1}$ and $r_{i+1}$ (if $i+1<n$) in which case $m_{ij}=3$. This means the Coxeter matrix can be easily built from this one. The Coxeter graph is the one for $n=7$ with a linear tail attached.
\par When wanting the subgroups generated by only swaps and $Tof$ or swaps and $Fred$ take the corresponding submatrix to keep only the desired generators. These are read off from the theorem before.

\end{theorem}

\begin{proof}
Write the permutations on $2^7$ that they induce and compute the orders of $r_i r_j$ for all pairs.

\end{proof}

\begin{proposition}
Assume $k \geq 3$. The remaining case are the groups $C^{k,n}$ generated by the $r_i=s_{i,i+1}$ and $r_n=C^{k,n}_{1\cdots k}$. For these the Coxeter matrix is the usual one for $S_n$ but with an extra row and column given as follows:

\begin{eqnarray*}
m_{n,n} &=& 1\\
m_{k,n} = m_{n,k} &=& 6\\
m_{i \neq k \; \text{or} \; n,n} = m_{n,i \neq k \; \text{or} \; n} &=& 2\\
\end{eqnarray*}

\end{proposition}

\begin{proof}
For $C^{k,n}$, the generators are the $s_{i,i+1}$ and $C^{k,n}_{1 \cdots k}$. The only $s_{i,i+1}$ that may fail to commute with this is $s_{k,k+1}$. That is the only extra edge in the Coxeter graph whose order needs to be computed using $k+1$ bits.

There are the following relevant cases for potential inputs.

\begin{itemize}
\setlength\itemsep{-1em}
\item $k-1$ all zeroes, k'th is zero and $k+1$'st is zero will go to possibility $6$ below under $r_k r_n$\\
\item $k-1$ all zeroes, k'th is zero and $k+1$'st is one will go to $8$ under $r_k r_n$\\
\item $k-1$ all zeroes, k'th is one and $k+1$'st is zero will go to $2$ under $r_k r_n$\\
\item $k-1$ all zeroes, k'th is one and $k+1$'st is one will go to itself under $r_k r_n$\\
\item $k-1$ all ones, k'th is zero and $k+1$'st is zero will go to itself under $r_k r_n$\\
\item $k-1$ all ones, k'th is zero and $k+1$'st is one will go to $7$ under $r_k r_n$\\
\item $k-1$ all ones, k'th is one and $k+1$'st is zero will go to $1$ under $r_k r_n$\\
\item $k-1$ all ones, k'th is one and $k+1$'st is one will go to $3$ under $r_k r_n$\\
\item If $k-1$ not uniform $k$'th and $k+1$'st are equal, then $r_n$ will have no influence and it will be fixed.\\
\item If $k-1$ not uniform $k$'th is zero and $k+1$'st is one, then $r_n$ will have no influence and it will go to the next possibility.\\
\item If $k-1$ not uniform $k$'th is one and $k+1$'st is zero, then $r_n$ will have no influence and it will go to the previous\\
\end{itemize}

So we see orbits or period $3$ and period $2$, therefore the total order is $6$. If $k=2$ then the possibility of the first $k-1$ not being uniform is impossible leaving order $3$. If $k=1$, then the order is $4$ by a direct check.

\end{proof}

\begin{cor}
All the Coxeter groups showing up this way in reversible computing are crystallographic.
\end{cor}

\begin{proof}
All the entries are in the set $\setof{1,2,3,4,6,\infty}$. We have not given the lattice, merely shown it exists.
\end{proof}

\section{3-SAT}

SAT-solvers have a long history and can be used for many problems. This is thanks to the Cook-Levin theorem. For example, one may encode graph isomorphism for two graphs of $N$ vertices into an instance of boolean satisfiability with $O(N^4)$ clauses to assure that every vertex of $G_1$ goes a vertex of $G_2$, that this is injective and that two vertices in $G_1$ connected by an edge go to a pair in $G_2$ connected by edge.

\begin{definition}[3-SAT]
Let $L$ be a logical formula $L \equiv C_1 \wedge \cdots C_m$ where each $C_m$ is a disjunction of $3$ literals in the variables $x_1 \cdots x_n$. The task is to take $L$ and see if there exists an assignment of the $x_i$ such that $L$ evaluates to TRUE.
\end{definition}

\begin{definition}[3-SAT Circuit]
A circuit $Circ_L$ on $n+1$ lines that takes $x_1 \cdots x_n$ as well as another bit $a$. The output is $x_1 \cdots x_n$ on the first $n$ lines and if $x_1 \cdots x_n$ is a valid assignment for the formula $L$, then $a \to \neg a$. Otherwise the output on the last line is unchanged.
\end{definition}

\begin{cor}
By the theorem above, there exists a number of auxiliary gates $k$ such that a circuit on $n+1+k$ lines can be made with only Toffoli gates and swaps such that when the ancilla are initialized to $0$ the output of the circuit is the output of the above circuit along with the ancilla set back to $0$. This guarantees that $Circ_L$ actually exists for any $L$.
\end{cor}

\begin{cor}
$Circ_L$ is now presented as a word in the Coxeter group $Tof^{n+1+k}$ for some $k$. If the word gets reduced to the identity we will know that there does not exist a valid assignment. However, that is too much to hope for because of the higher order relations. All we can guarantee is that the higher order relations are needed to witness that this word is actually the identity in $S_{2^{n+1+k}} \subset U(2^{n+1+k})$.\\
\end{cor}

We have given the existence of the circuit $Circ_L$, but we can be more concrete and provide an explicit realization.

\begin{lemma}
For one clause such as $L=x_i \vee x_j \vee x_k$, the circuit $Circ_L$ has an expression with $\leq 24$ gates and $n+1$ main input bit lines and $2$ ancilla.
\end{lemma}

\begin{proof}
For a general choice of $x_i, x_j , x_k$ from the $n$ possibilities, we must conjugate by some permutation of $n$. Use up to $3$ swaps as necessary. If the clause is of the form $\neg x_1 \vee x_2 \vee x_3$ conjugate by $NOT_1$ and mutatis munandis for the others. If all of these are used then that is a conjugation by $6$ gates on each side of the circuit.

Now without loss of generality, let $n=3$ and the clause be $x_1 \vee x_2 \vee x_3$. This circuit can be constructed by first applying a not to $x_{1,2,3}$, then apply a $Tof_{1,2,5}$ where $5$ is an ancilla. Then apply a $Tof_{3,5,6}$ with another ancilla $6$. Then apply a $CNOT_{6,4}$. Then apply a $NOT_4$. This takes care of the desired behavior on the $a$ bit which is denoted $4$ here. From here the ancilla's and input's are restored to their values by reversing the previous modifications in backwards order.

So for a formula made up of this single clause we have $\abs{Circ_L} \leq 12$ where $\abs{}$ indicates the number of generators used in the word.
\end{proof}

\begin{lemma}
For the case $x_1 \vee x_2 \vee x_3$, $Circ_L = NOT_1 NOT_2 NOT_3 Tof_{1,2,5} Tof_{3,5,6} NOT_4 CNOT_{6,4} Tof_{3,5,6} Tof_{1,2,5} NOT_3 NOT_2 NOT_1$. The Coxeter group using these generators as follows has structure

\begin{itemize}
\setlength\itemsep{-1em}
\item $r_1 = NOT_1$\\
\item $r_2 = NOT_2$\\
\item $r_3 = NOT_3$\\
\item $r_4 = Tof_{1,2,5}$\\
\item $r_5 = Tof_{3,5,6}$\\
\item $r_6 = NOT_4$\\
\item $r_7 = CNOT_{6,4}$\\
\end{itemize}

\begin{eqnarray*}
Circ_L &=& r_1 r_2 r_3 r_4 r_5 r_6 r_7 r_5 r_4 r_3 r_2 r_1
\end{eqnarray*}

Filling in the entries already derived from the previous section gives:

\begin{eqnarray*}
M &=& \begin{pmatrix}
1 & 2 & 2 & ?_1 & 2 & 2 & 2\\
2 & 1 & 2 & ?_1 & 2 & 2 & 2\\
2 & 2 & 1 & 2&  ?_1 & 2 & 2\\
?_1 & ?_1 & 2 & 1 & ?_2 & 2 & 2\\
2 & 2 & ?_1 & ?_2 & 1 & 2 & ?_3\\
2 & 2 & 2 & 2 & 2 & 1 & ?_4\\
2 & 2 & 2 & 2 & ?_3 & ?_4 & 1
\end{pmatrix}\\
\end{eqnarray*}

The remaining entries are

\begin{eqnarray*}
( NOT_1 Tof_{1,2,5} )^{?_1} = (NOT_2 Tof_{1,2,5})^{?_1} &=& (NOT_3 Tof_{3,5,6})^{?_1} = e\\
( Tof_{1,2,5} Tof_{3,5,6} )^{?_2} &=& e\\
( Tof_{3,5,6} CNOT_{6,4} )^{?_3} &=& e\\
( NOT_4 CNOT_{6,4} )^{?_4} &=& e\\
?_1 &=& 4\\
?_2 &=& 4\\
?_3 &=& 4\\
?_4 &=& 2\\
\end{eqnarray*}

In this case, the word is already reduced for this Coxeter group. The Coxeter compiler leaves the circuit as is.

\end{lemma}

\begin{proof}

$r_6$ commutes with everything else, so bring it out front. Then create the Coxeter graph by specifying all the entries of the Coxeter matrix that do not have entry either $1$ or $2$.

\begin{lstlisting}[label=SingleClause,caption=Single Clause in Sage]
G = Graph([(1,4,4), (2,4,4), (3,5,4), (4,5,4), (5,7,4)])
W = CoxeterGroup(G)
CoxeterGroup(W.coxeter_diagram()) is W
s = W.simple_reflections()
w = s[1]*s[2]*s[3]*s[4]*s[5]*s[7]*s[5]*s[4]*s[3]*s[2]*s[1]
w.coset_representative([]).reduced_word()
\end{lstlisting}

This outputs $[2,1,4,3,5,7,5,4,3,2,1]$ which is a simple reordering because $r_{1,2}$ and $r_{3,4}$ are commuting pairs.

\end{proof}

\begin{remark}
For a case where it actually does something run the following instead with the same $G$, $W$ and $s$. Note the speeds and reductions of lengths.

\begin{lstlisting}[label=RandomExample,caption=Randomly Generated Examples]
wHelper=[choice(list([1,2,3,4,5,7])) for i in range(50)]
wHelper
w=s[wHelper[0]]
for i in wHelper[1:]:
    w=w*s[i]
w.coset_representative([]).reduced_word()
\end{lstlisting}

\end{remark}

\begin{lemma}
Suppose we have the circuits $Circ_{L_1}$ and $Circ_{L_2}$, then if $L = L_1 \wedge L_2$, $Circ_L$ can be built with length $\leq 5 + \abs{Circ_{L_1}} + \abs{Circ_{L_2}}$ and $n+1$ main input bit lines and $4$ additional ancilla from those for $L_1$ and $L_2$. If one gives up the desire to be able to run $Circ_{L_1}$ and $Circ_{L_2}$ in parallel, then those ancilla can be reused because they have been reset by the time the second one starts.
\end{lemma}

\begin{proof}
Let $a_1$ be the target bit used for $Circ_{L_1}$ and $a_2$ for that on $Circ_{L_2}$. These will be ancillas for $Circ_L$. Let there be two new ancilla's $b_1$ and $b_2$ as well. First apply a $CNOT_{a_1,b_1}$ and $CNOT_{a_2,b_2}$, then apply $Circ_{L_1}$ and $Circ_{L_2}$. $a_i$ and $b_i$ will be opposite if $Circ_{L_i}$ was satisfied. Applying a $CNOT_{a_i,b_i}$ now will ensure the value of $b_i$ indicates whether or not $Circ_{L_i}$ was satisfied. Now apply a $Tof_{b_1,b_2,a_3}$ where $a_3$ is the bit for $Circ_L$ which flips upon satisfiable assignments for all $L$. This gives length $\leq 5 + \abs{Circ_{L_1}} + \abs{Circ_{L_2}}$.
\end{proof}

\begin{proposition}
For a formula $L$ with $m$ clauses, the length of the word constructed in this manner is $\leq 5m + \sum_j \abs{Circ_{L_j}}$. This is bounded above by $29m$. The number of ancilla is $\leq 4 m + \sum_{j \leq m} \abs{Anc_{L_j}}$ which is bounded above by $6m$. If the $Circ_{L_i}$ are done in series so their ancilla are reused, then one can get away with $\leq 4 m + 2$.
\end{proposition}

\begin{proof}
Let $Circ_{L,\leq j}$ be the result of only the first $j$ clauses and $Circ_{L,j}$ the result of only clause $j$.

\begin{eqnarray*}
\abs{Circ_L} &\leq& 5 + \abs{Circ_{L,\leq m-1}} + \abs{Circ_{L_m}}\\
&\leq& 5 + (5(m-1) + \sum_{j \leq m-1}  \abs{Circ_{L_j}} ) + \abs{Circ_{L_m}}\\
&\leq& 5 m + \sum_j \abs{Circ_{L_j}}
\end{eqnarray*}

The base case of $m=1$ gives

\begin{eqnarray*}
\abs{Circ_{L_1}} &\leq& 5 + \abs{Circ_{L_1}}
\end{eqnarray*}

For the number of ancilla:

\begin{eqnarray*}
\abs{Anc_L} &\leq& 4 + \abs{Anc_{L,\leq m-1}} + \abs{Anc_{L_m}}\\
&\leq& 4 + ( 4 (m-1) + \sum_{j \leq m-1} \abs{Anc_{L_j}} ) + \abs{Anc_{L_m}}\\
&\leq& 4 m + \sum_{j \leq m} \abs{Anc_{L_j}}\\
\abs{Anc_{L_1}} &\leq& 4 + \abs{Anc_{L_1}}\\
\end{eqnarray*}

\end{proof}

\begin{definition}[RevId]
Let $C$ be a depth $m$ circuit on $n$ bits with gates drawn from $NOT_i$, $CNOT_{ij}$, $TOF_{ijk}$, $SWAP_{ij}$. Determine whether or not the circuit is the identity. Specifying the gate is a choice of $n+n(n-1)+\frac{n(n-1)(n-2)}{2}+\frac{n(n-1)}{2}$ possibilities. Raising this to the $m$th power gives the number of possibilities for the input.

Let RevId be the decision problem of telling whether or not this will reduce to the identity. This generalizes 3SAT by allowing more arbitrary circuits instead of just of the form $Circ_L$. The first pass through simplifying this circuit will be efficient by the Coxeter compiler. The higher order relations make this decision problem harder.
\end{definition}

\begin{example}
For $Circ_L$ with $N$ variables and $M$ clauses, there are $\leq N+1+6M$ bits including the ancilla and depth $\leq 29 M$. Plug those estimates in for $n$ and $m$.
\end{example}

This described the decision problem of whether there was a satisfying assignment. It also did so in a straightforward way. The Coxeter compiler should be used to reduce the depth. There is also the problem of finding those satisfying assumptions \cite{TodaSoh}.

\section{Quantum Gates}

The following two theorems motivate the choice of gate sets in the quantum case.

\begin{theorem}[\cite{Selinger2}]
A unitary matrix with entries in $\mathbb{Z}[\frac{1}{\sqrt{2}},i]$ has an exact representation over the Clifford+T gate set with possibly one ancilla.
\end{theorem}

\begin{theorem}[\cite{Selinger}]
Arbitrary single qubit rotations are efficiently approximated within the Clifford+T gate set. In fact this is done more efficiently than that guaranteed by Solovay-Kitaev thanks to number theoretic structure.
\end{theorem}

If these algorithms or Solovay-Kitaev are used to approximate circuit elements, when the building blocks get put together to make more complicated circuits we will have a word in the Clifford+T gate set. It will likely no longer be reduced, as some circuit elements may undo parts of others. But this gives us gates that have order $4$ and $8$, like the phase gate $S$ and the eponymous $T$. So in order to fit within the framework we have built, let us only pick involutions that are in this group to make as generators. Treat the others as stuck in place for now. We will get some subgroups of the full Cliff+T groups which are dense subgroups of $U(2^n)$. For small values of $n$, similar problems have been previously studied in \cite{Planat}.

\begin{prop}
Let $X_a , Y_a , Z_a , H_a$ be the single qubit gates acting on qubit $a$ among a total of $n$. Let $S_{a,a+1}$ be the swap for $a,a+1$. The group generated by these is a quotient of the following Coxeter group with generators $S_{a,a+1}$, $X_1$, $Y_1$, $Z_1$ and $H_1$. Index these as $r_1 \cdots r_{n-1}$ for the symmetric group part and $r_n \cdots r_{n+3}$ for the single qubit gates.

\begin{eqnarray*}
m_{ij} &=& \begin{cases}
3 & i \leq n-1 , j \leq n-1 , \abs{i-j}=1\\
4 & \setof{i,j}=\setof{n,n+1}\\
4 & \setof{i,j}=\setof{n,n+2}\\
8 & \setof{i,j}=\setof{n,n+3}\\
4 & \setof{i,j}=\setof{n+1,n+2}\\
4 & \setof{i,j}=\setof{n+1,n+3}\\
8 & \setof{i,j}=\setof{n+2,n+3}\\
4 & \setof{i,j}=\setof{1,n}\\
4 & \setof{i,j}=\setof{1,n+1}\\
4 & \setof{i,j}=\setof{1,n+2}\\
4 & \setof{i,j}=\setof{1,n+3}\\
\end{cases}
\end{eqnarray*}

\end{prop}

\begin{proof}

A single qubit gate acting on $a \neq 1$ can be acheived through conjugation with $(1,a)$ in the symmetric group generated by $S_{a,a+1}$. It remains to compute the values $m_{ij}$ which are given by the computations below:

\begin{eqnarray*}
( X_1 Y_1 )^4 &=& 1\\
( X_1 Z_1 )^4 &=& 1\\
( Y_1 Z_1 )^4 &=& 1\\
( H_1 X_1 )^8 &=& 1\\
( H_1 Y_1 )^4 &=& 1\\
( H_1 Z_1 )^8 &=& 1\\
(H_1 S_{1,2})^4 &=& 1\\
(X_1 S_{1,2})^4 &=& 1\\
(Y_1 S_{1,2})^4 &=& 1\\
(Z_1 S_{1,2})^4 &=& 1\\
(H_{2} S_{1,2})^4 &=& 1\\
(X_{2} S_{1,2})^4 &=& 1\\
(Y_{2} S_{1,2})^4 &=& 1\\
(Z_{2} S_{1,2})^4 &=& 1\\
(S_{a,a+1} S_{a+1,a+2})^3 &=& 1\\
\end{eqnarray*}
\end{proof}

\begin{prop}
Let $n \geq 4$. Adjoin $cX_{12}$, $cY_{12}$, $cZ_{12}$, $cH_{12}$, $Tof_{123}$ and $Fred_{123}$ as $r_{n+4} \cdots r_{n+9}$. $cX_{12}$ stands for the controlled version with the control being $1$ and operating nontrivially on $2$. Similarly for the others. $Tof$ and $Fred$ come from the inclusion of reversible into quantum computing. The Coxeter matrix for $n=4$ is 

\begin{eqnarray*}
M &=& \left(
\begin{array}{ccccccccccccc}
 1 & 3 & 2 & 4 & 4 & 4 & 4 & 3 & 3 & 2 & >8  & 2 & 6 \\
 3 & 1 & 3 & 2 & 2 & 2 & 2 & 4 & 4 & 4 & 4 & 6 & 2 \\
 2 & 3 & 1 & 2 & 2 & 2 & 2 & 2 & 2 & 2 & 2 & 4 & 6 \\
 4 & 2 & 2 & 1 & 4 & 4 & 8 & 4 & 4 & 4 & 4 & 4 & 4 \\
 4 & 2 & 2 & 4 & 1 & 4 & 4 & 4 & 4 & 4 & 4 & 4 & 4 \\
 4 & 2 & 2 & 4 & 4 & 1 & 8 & 2 & 2 & 2 & 2 & 2 & 2 \\
 4 & 2 & 2 & 8 & 4 & 8 & 1 & 8 & 8 & 8 & 8  & 8 & 8 \\
 3 & 4 & 2 & 4 & 4 & 2 & 8 & 1 & 4 & 4 & 8 & 4 & 4 \\
 3 & 4 & 2 & 4 & 4 & 2 & 8 & 4 & 1 & 4 & 4 & 4 & 4 \\
 2 & 4 & 2 & 4 & 4 & 2 & 8 & 4 & 4 & 1 & 8 & 2 & 4 \\
 >8  & 4 & 2 & 4 & 4 & 2 & 8  & 8 & 4 & 8 & 1 & 8 & 4 \\
 2 & 6 & 4 & 4 & 4 & 2 & 8 & 4 & 4 & 2 & 8 & 1 & 3 \\
 6 & 2 & 6 & 4 & 4 & 2 & 8 & 4 & 4 & 4 & 4 & 3 & 1 \\
\end{array}
\right)
\end{eqnarray*}

The Coxeter matrix for general $n$ has $2$ for the entries corresponding to $s_{a,a+1}$ for $a \geq 4$ and any of the other generators except $s_{a+1,a+2}$ and $s_{a-1,a}$. That is to attach a linear tail to the Coxeter graph.

\end{prop}

\begin{proof}
These are also involutions. Again it remains to compute the order of their products. $>8$ indicates that the order of the corresponding $r_i r_j$ was $>8$. This corresponds to $s_{12} \cdot cH$ which has eigenvalues

\begin{eqnarray*}
Spec(s_{12} \cdot cH) &=& \setof{\frac{1}{4} \left(-2-\sqrt{2}+i \sqrt{16-\left(2+\sqrt{2}\right)^2}\right), \frac{1}{4} \left(-2-\sqrt{2}-i \sqrt{16-\left(2+\sqrt{2}\right)^2}\right),1}\\
\end{eqnarray*}
with multiplicities $4$, $4$ and $8$ respectively.

These are on the unit circle as they must be, but one must determine whether they are roots of unity and their orders in order to fix this entry of the Coxeter matrix. The minimal polynomial for this is $2x^4 + 4 x^3 + 5 x^2 + 4 x + 2$ so the value in concern is not an algebraic integer. Therefore the $>8$ gets replaced by $\infty$.

\end{proof}

As with the classical reversible case, take whichever generators to produce the Coxeter group whose quotient by higher order relations is desired. Or for a different generating set $\tilde{r}_i$ (also involutions), the new Coxeter matrix can be produced from this one by expanding the new generators in the old generators $\tilde{r}_i = r_{i_1} \cdots r_{i_k}$ and using this matrix to evaluate all the orders of $\tilde{r}_i \tilde{r}_j$ without doing the computation in $U(2^n)$ again.

\begin{lemma}
The group generated by all of the above except $cH_{12}$ has Serre's property FA.
\end{lemma}

\begin{proof}
The Coxeter matrix without that generator has all entries finite. That means the associated Coxeter group has property FA  \cite{SerreTrees}. The group of concern is then a quotient by some group normalizing the group generated by those higher order relations. Property FA is preserved under quotients.
\end{proof}

\section{Rewrite Algorithm}

Word problems for groups host a wide variety of behaviors, but can be well behaved for certain classes of groups\cite{HermillerShapiro,GrovesSmith1,GrovesSmith2}. Coxeter groups are one such class which have good word problems as studied in \cite{Hermiller,BrinkHowlett,Edelman,CoxeterBook}. This was the main motivation for this work. The constrained setting of only the gates generated by involutions allows one to write efficient programs.

Input a word $w$ in any of the Coxeter groups above as well as the associated Coxeter graph.

\begin{definition}[Dependence DAG]
Given a word of computational gates, build a DAG as follows. Start with a vertex called START. Put the first gate as a successor to START. If the next gate commutes with the first gate, put it as a successor to START. If not, it is a successor to the first gate. This continues with the only information about commutativity to decide the edges indicating potential dependency. This construction uses only the connectivity of the Coxeter matrix (ignores the edge weights on the Coxeter graph).
\end{definition}

Pick a subset $S$ of the vertices of the dependence DAG, let $H(S)$ be the subset of vertices of $S$ and all it's predecessors. Do a topological sort on the induced subgraph on $H(S)$ to give a new word in the Coxeter group which can be viewed as a prefix of $w$. Do the same for the rest which is the corresponding suffix. For these subsets also restrict to the subset of the Coxeter graph using only the generators that show up in $H(S)$. Reducing these shorter words within their own Coxeter groups is more manageable. For example, let $S$ be the vertices that come from the middle third $w_2$ of the word $w=w_1 w_2 w_3$ where each $w_i$ has approximately the same length. This choice ensures that $\frac{1}{3} \leq \frac{\abs{H(S)}}{\ell(w)} \leq \frac{2}{3}$.

Another tool to speed up the solution to the word problem is the intervening words property.

\begin{theorem}[Intervening Neighbors Property \cite{Intervening}]
Let $C$ be an infinite irreducible Coxeter group. A word in $C$ has the intervening neighbors property if for any two occurences of a generator $r_i$ all the neighbors $r_j$ of $r_i$ show up in between these two.
\par All words with the intervening neightbors property are irreducible.
\end{theorem}

So either we find the word is irreducible and we are finished or we find an $i$ such that there is a segment $r_i \cdots r_i$ without all it's neighbors. This gives a certain potentially reducible word to focus attention on. Given an algorithm to reduce these segments, this theorem speeds up reduction for infinite irreducible Coxeter groups by focusing attention on smaller segments. Infinite irreducibility can be judged directly from the Coxeter matrix.

\subsection{Additional Relations}

As stated earlier there may be additional relations. However, if we have established that a word represents the identity with only the basic Coxeter relations, we are assured that it is the identity.

After using the $r_i^2=e$ relations, all the relations that involve $2$ distinct generators must be of the form $e=r_i r_j r_i \cdots r_j$. So these are all already taken care of. The new relations involve $3$ or more distinct generators. Call the $R_{g}$ for $g$ the number of distinct generators and $R$ for all relations. A relation among $N$ generators has potential to be new information only if there exists no bipartitioning of $N = A \bigcup B$ such that $r_i$ for $i \in A$ and $r_j$ for $j \in B$ all have disjoint supports.

\begin{lemma}
One may restrict to the case of 7 (qu)bits in order to determine the new relations which involve $3$ distinct generators.
\end{lemma}

\begin{proof}
Each $r_i$ from the gate set operates on at most $3$ lines at once. In order for the supports not to have a bipartitioning, all together they must operate on at most $7$ qu(bits). Without loss of generality, rename these to be $1 \cdots 7$. Let $r_1 \cdots r_{M(7)}$ be those gates for $n=7$.

Generate words in exactly $3$ letters from this alphabet. Make sure they have overlapping supports. Reduce them using only the Coxeter relations. If they don't give the empty word already, check whether they give the identity in $S_{2^7}$ or $U(2^7)$ as appropriate. Keep generating longer words until such a relation is found. The relations this spawns by renaming all the (qu)bits to one of $1 \cdots n$ are the additional relations we are looking for. Let $\bar{R}_3$ be the set consisting of these as well as their cyclic permutations and their inverses. From this, one can perform Dehn's greedy algorithm on $\langle r_1 \cdots r_{M(n)} \mid \bar{R}_3 \rangle$.
\end{proof}

\begin{figure}[htb!]
\centering
\begin{tikzcd}
\langle r_1 \cdots r_M | R_2 \rangle \arrow[d]\\
\langle r_1 \cdots r_M | R_2 , R_3 \rangle \arrow[d]\\
\langle r_1 \cdots r_M | R_2, R_3, R_4 \rangle \arrow[d]\\
\langle r_1 \cdots r_M | R \rangle\\
\end{tikzcd}
\caption{We seek to reduce words in $\langle r_1 \cdots r_M | R \rangle$ which is a subgroup of either $S_{2^n}$ or $U(2^n)$. We do this by successive approximation, by staying at the top level as much as possible.}
\end{figure}

For the 3-SAT problem, this is especially useful, in that it gives a successive approximation of the semantics such that only when we have imposed all the relations in the reversible computing group, have we solved the problem. But with this filtration on the relations by number of generators involved, we may solve the word problem in gradually finer approximation which gradually get harder. The hope is that for the average case only a relatively course approximation is enough to solve it's satisfiability while the worst case is responsible for the NP hardness of the problem.

There are other techniques that can be combined with the Coxeter perspective in order to hit the problem with multiple hammers. For reversible circuits, for example, one can do a dynamic programming approach \cite{Shende}.

\section{3-SAT Redux}

If we have the circuit for the reversible computing, we may promote it to the quantum case and then reduce. This will allow the intermediate steps to take advantage of the quantum resources while still maintaining the overall computation. A prominent example is making quantum circuits that identify certain properties of SHA-256 (like first 30 bits are 0's) that are fed into Grover search algorithm as $U_f$. In this section, let us continue from the example of $Circ_L$.

\begin{lemma}[Operator Identity Circuit]

Let $C$ be a circuit that we wish to identify whether or not it is the identity operating on $n+1+\abs{Anc_L}$ number of qubits. Break it up into two subwords $C=AB$. Then the question becomes whether or not $A W=B^{-1} W$ where $W$ is an arbitrary circuit of polynomial depth.

Use these circuits to create the state $\ket{input} = (\ket{0} \otimes A_1 W_1 \ket{0 \cdots 0}  \otimes B_1^{-1} W_1 \ket{0 \cdots 0}) \otimes (\ket{0} \otimes A_2 W_2 \ket{0 \cdots 0}  \otimes B^{-1}_2 W_2 \ket{0 \cdots 0}) \cdots (\ket{0} \otimes A_k W_k \ket{0 \cdots 0}  \otimes B_k^{-1} W_k \ket{0 \cdots 0}) \otimes \ket{0}$. $C=A_i B_i$ are all divisions of the word $C$. They may or may not be chosen to be the same. The $W_i$ can be chosen independently or they might all be the same.

\begin{eqnarray*}
\ket{\phi_i} &\equiv& A_i W_i \ket{0 \cdots 0}\\
\ket{\psi_i} &\equiv& B^{-1}_i W_i \ket{0 \cdots 0}\\
SWAPTEST^{\otimes k} \otimes Id_2 \ket{input} &=& \bigotimes_{i=1}^k \bigg( \ket{0} \otimes (\frac{1}{2} \ket{\phi_i} \otimes \ket{\psi_i} + \frac{1}{2} \ket{\psi_i} \otimes \ket{\phi_i}) \\ &+& \ket{1} \otimes (\frac{1}{2} \ket{\psi_i} \otimes \ket{\phi_i} - \frac{1}{2} \ket{\phi_i} \otimes \ket{\psi_i}) \bigg) \otimes \ket{0}\\
\end{eqnarray*}

\begin{eqnarray*}
UCombiner \ket{e_1 \cdots e_k} \otimes \ket{a} &=& \ket{e_1 \cdots e_k} \otimes \ket{a + f(e_1 \cdots e_k)}
\end{eqnarray*}

where $f(e_1 \cdots e_k) = 1$ if and only if all the $e_i$ are $0$. Otherwise it is $0$. If desired do a permutation so we can write the full operator which can be applied as $Id \otimes UCombiner$. By abuse of notation we will just denote the full operator applied on the correct indices as $U$.

\begin{eqnarray*}
UCombiner \circ (SWAPTEST^{\otimes k} \otimes Id_2) \ket{input} &=& \ket{B} \otimes \ket{1} + \ket{A} \otimes \ket{0}\\
\end{eqnarray*}

where 

\begin{eqnarray*}
\ket{B} &\equiv& \bigotimes_{i=1}^k \bigg( \ket{0} \otimes (\frac{1}{2} \ket{\phi_i} \otimes \ket{\psi_i} + \frac{1}{2} \ket{\psi_i} \otimes \ket{\phi_i})  \bigg)
\end{eqnarray*}

and $\ket{A}$ is some other vector on $\mathbb{C}^{(2 \cdot 2^{n+1+\abs{Anc_L}} \cdot 2^{n+1+\abs{Anc_L}})^k}$ which is unnecessary to expand out. But it is orthogonal to $\ket{B}$. The probability of observing $1$ on the final answer qubit is 

\begin{eqnarray*}
Prob(1) &=& \prod_{i=1}^k \frac{1}{4} (2 + 2 \abs{\bra{0 \cdots 0} W_i^\dagger A_i^\dagger B_i^{-1} W_i \ket{0 \cdots 0} }^2)\\
&=& \frac{1}{2^k} \prod_{i=1}^k (1 + \abs{\bra{0 \cdots 0} W_i^\dagger A_i^\dagger B_i^{-1} W_i \ket{0 \cdots 0} }^2)\\
&=& \frac{1}{2^k} \prod_{i=1}^k (1 + \abs{ F_i }^2)\\
F_i &\equiv& \bra{0 \cdots 0} W_i^\dagger A_i^\dagger B_i^{-1} W_i \ket{0 \cdots 0}\\
&=& Tr ( \ket{0 \cdots 0} \bra{0 \cdots 0} W_i^\dagger A_i^\dagger B_i^{-1} W_i )\\
\end{eqnarray*}

An amplitude amplification can be inserted $N$ times.

\begin{eqnarray*}
\ket{B} \otimes \ket{1} + \ket{A} \otimes \ket{0} &=& \sqrt{\bra{B}\ket{B}} \frac{\ket{B}}{\sqrt{\bra{B} \ket{B}}} \otimes \ket{1} + \sqrt{\bra{A}\ket{A}} \frac{\ket{A}}{\sqrt{\bra{A} \ket{A}}} \otimes \ket{0}\\
\theta &\equiv& \cos^{-1} \sqrt{Prob(1)}\\
Q^N ( \ket{B} \otimes \ket{1} + \ket{A} \otimes \ket{0} ) &=& \cos (2N+1) \theta \cdot \frac{\ket{B}}{\sqrt{\bra{B} \ket{B}}} \otimes \ket{1} + \sin (2N+1) \theta \cdot \frac{\ket{A}}{\sqrt{\bra{A} \ket{A}}} \otimes \ket{0}\\
Prob_{ampd} (1) &=& \cos^2 (2N+1) \cos^{-1} \sqrt{Prob(1)} = T_{2N+1}^2 (\sqrt{Prob(1)})\\
Prob(1) = 1 - \epsilon &\implies& Prob_{ampd}(1)=1 - (2N+1)^2 \epsilon + O(\epsilon^2)\\
\end{eqnarray*}

\end{lemma}

\begin{lemma}[3-SAT]
Let $C=Circ_L$, then $F_i$ has the form used in the Harish-Chandra-Itzykson-Zuber integral. $A_i^\dagger B_i^\dagger = 1 - 2 P_{-1}$ where $P_{-1}$ is a projector onto a subspace of dimension $2^{n+\abs{Anc_L}} PV$ which is exponentially small relative to the full $2^{n+1+\abs{Anc_L}}$. $F_i = 1 - 2 \bra{R} P_{-1} \ket{R}$ where $R$ is a random state.

If we choose $A_i = Circ_L$, $B_i = Id$ and $W_i \ket{0 \cdots 0} = \frac{1}{\sqrt{2^n}} \sum_{i_1 \cdots i_n = 0,1} \ket{i_1 \cdots i_n} \otimes \ket{0} \otimes \ket{0}^{\abs{Anc_L}} $, then we may give explicit values for the probabilities in terms of the number of valid assignments.

\begin{eqnarray*}
F_i &=& \frac{1}{2^n} \abs{ \text{invalid}_L } = PI\\
Prob(1) &=& \frac{1}{2^k} (1 + PI^2)^k\\
\end{eqnarray*}

where $PI$ is the probability a random assignment will not satisfy $L$. Giving an upper bound to this, provides an upper bound to $Prob(1)$ which is the probability for falsely claiming that $Circ_L = Id$. As $n \to \infty$, the upper bound for $PI$ goes to $1$. This means that the required $k$ to keep this false negative $\leq \frac{1}{3}$, blows up exponentially. Therefore this does not provide a proof of $NP \subset BQP$.

If we let the $W_i$ be words in $S_{a,a+1}$, $X_1$, $Y_1$, $Z_1$, $H_1$, $cX_{12}$, $cY_{12}$, $cZ_{12}$ $cH_{12}$, $Tof_{123}$ and $Fred_{123}$ (or some other set of generating involutions), then this will be a quite large circuit that should be fed into the Coxeter compiler before being executed.

\end{lemma}

\begin{proof}

\begin{eqnarray*}
Circ_L^2 = Id &\implies& A_i B_i A_i B_i = Id\\
B_i^{-1} A_i^{-1} &=& A_i B_i\\
A_i^\dagger B_i^\dagger A_i^\dagger B_i^\dagger &=& A_i^\dagger A_i B_i B_i^\dagger = Id\\
(A_i^\dagger B_i^\dagger)^{-1} &=& (A_i^\dagger B_i^\dagger) = (A_i^\dagger B_i^\dagger)^{\dagger}
\end{eqnarray*}

Therefore both $\ket{0 \cdots 0} \bra{0 \cdots 0}$ and $A_i^\dagger B_i^\dagger$ are Hermitian. The Harish-Chandra-Itzykson-Zuber formula uses the spectrum of both the Hermitian matrices. The first is a projector onto $\ket{0 \cdots 0}$ so has a single $1$ and the rest $0$. The spectrum of $A_i^\dagger B_i^\dagger$ is:

\begin{eqnarray*}
Spec ( A_i^\dagger B_i^\dagger ) &=& Spec (A_i A_i^\dagger B_i^\dagger A_i^\dagger)\\
&=& Spec (B_i^\dagger A_i^\dagger)\\
&=& Spec (A_i B_i) = Spec (Circ_L)\\
\end{eqnarray*}

which is $\pm 1$ with multiplicities that depend on $PI$.

Also note that the problem of $Circ_L^2 = Id$ and we wish to test whether $Circ_L = Id$ also implies the same for $A_i^\dagger B_i^\dagger$. This means that instead of just dividing $Circ_L$ into two pieces, we can repeat this process and shuffle more drastically. These are similar to the original $Circ_L$ so if $W_i$ is truly uniformly drawn, this does not make a difference.

\begin{eqnarray*}
A_i^\dagger B_i^\dagger A_i^\dagger B_i^\dagger &=& A_i^\dagger A_i B_i B_i^\dagger\\
A_i B_i = Id &\implies& A_i^\dagger B_i^\dagger = Id\\
\end{eqnarray*}

\end{proof}

\section{Expander Graphs}

Consider the Coxeter graph. If we divide the vertices of this graph into two subsets $A$ and $A^c$ and have the compiler take care of those individual parts, the amount of interaction between these two is measured heuristically by the number of edges connected them (without the weighting of the edges). In addition to have this be a good parallel division, we would like to make those $A$ and $A^c$ as close to half the vertices as possible. This leads to the idea that the Cheeger constant of the Coxeter graph is a heuristic for good parallelizations.

\begin{definition}[Cheeger Constant]
\begin{eqnarray*}
h(G) &=& min \frac{E(A,A^c)}{min(\abs{A},\abs{A^c})}
\end{eqnarray*}

where $E(A,A^c)$ counts edges connecting $A$ and $A^c$. Let there be a family of graphs $G_n$. If for all $n$ $h(G_n) \geq c > 0$ then that is called an expanding family.

\end{definition}

Note this measuring the expansion property on a graph where the vertices are operations not the qubits themselves. One can phrase this generally as letting $H$ be the multi-hypergraph with qubits as vertices and possible gates as hyperedges. Then one constructs a graph $G_H$ by letting the vertices be indexed by the hyperedges. These are not connected if and only if the associated operations commute. Similar questions of constructing the line graph of a hypergraph and asking about it's expansion properties are addressed in \cite{Badaoui} (thanks to \cite{Arnaud}).

\section{Conclusion}

For certain sets of both reversible and quantum gates we have presented the Coxeter group generated by imposing only the relations of the form $(r_i r_j)^{m_{ij}}$. This gives groups with good rewrite properties in order to make compilation of circuits more efficient. This procedure is also easily generalizable by taking any set of gates that are involutions and forming a corresponding Coxeter group as a coarse approximation of the semantics.

As a particular example, we gave a circuit for 3SAT as well as identifying whether or not that circuit was the identity. This does not give a proof of $NP \subset BQP$ because we have not gotten the probability of false negatives below $\frac{1}{3}$. Such a proof would imply a collapse of the polynomial hierarchy to $\Sigma_2^P$ \cite{BQPPH}. The word problem in groups satisfying a perimetric condition is in $NP$ \cite{BirgetOlshanskiiRipsSapir}. $NP \subset BQP$ would imply a quantum computer would have an algorithm $A$ that would input a program $P$ to make sure it was told to do was doing useful work. Such a subroutine would go into a truly quantum optimizing compiler. In addition, if there exists a $\Sigma_2^P$ complete problem not in $BQP$, $NP \subset BQP$ would imply $P \neq NP$.

A manageable way to expand this would be to include generators of orders $4$ or $8$ in order to handle other gates in the Clifford+T groups such as the $S$ and $T$. We hope the rewrite programs for the Coxeter subgroups can be used to aid the rewrite programs for the full Clifford+T group. If $K$ and $Q$ have finite complete rewriting systems, then a $G$ which fits in between with a short exact sequence will also have one \cite{GrovesSmith1}. This property will be useful in developing the compiler for larger gate sets.

If this algorithm was sufficiently efficient, the user of such a device would only have to input a very naive sequence of gates to solve the problem and then let the word problem take care of the rest. We could even bootstrap this compiler by having the subroutines of finding subgraphs in the computational DAG into quantum programs themselves. This involves combining with techniques from \cite{QuantumDynProgDAG,GiacomoPark,VenturelliDoRieffelFrank} to improve the scheduling of gates.

\bibliographystyle{unsrt}
\bibliography{CoxeterCompiler}

\begin{thebibliography}{10}

\bibitem{Geordie}
Ben Elias and Geordie Williamson.
\newblock Soergel calculus.
\newblock {\em Representation Theory of the American Mathematical Society},
  20(12):295--374, 2016.

\bibitem{Stay}
John Baez and Mike Stay.
\newblock Physics, topology, logic and computation: a rosetta stone.
\newblock In {\em New structures for physics}, pages 95--172. Springer, 2010.

\bibitem{KashaevKolya}
RM~Kashaev and N~Reshetikhin.
\newblock Affine toda field theory as a 3-dimensional integrable system.
\newblock {\em Communications in mathematical physics}, 188(2):251--266, 1997.

\bibitem{Aaronson}
S.~{Aaronson}, D.~{Grier}, and L.~{Schaeffer}.
\newblock {The Classification of Reversible Bit Operations}.
\newblock {\em ArXiv e-prints}, April 2015.

\bibitem{CockettComfort}
Cole Comfort and J.~Robin~B. Cockett.
\newblock The category {TOF}.
\newblock {\em CoRR}, abs/1804.10360, 2018.

\bibitem{TodaSoh}
Takahisa Toda and Takehide Soh.
\newblock Implementing efficient all solutions sat solvers.
\newblock {\em Journal of Experimental Algorithmics (JEA)}, 21:1--12, 2016.

\bibitem{Selinger2}
Brett Giles and Peter Selinger.
\newblock Exact synthesis of multiqubit clifford+ t circuits.
\newblock {\em Physical Review A}, 87(3), 2013.

\bibitem{Selinger}
N.~J. {Ross} and P.~{Selinger}.
\newblock {Optimal ancilla-free Clifford+T approximation of z-rotations}.
\newblock {\em ArXiv e-prints}, March 2014.

\bibitem{Planat}
Michel Planat and Maurice Kibler.
\newblock Unitary reflection groups for quantum fault tolerance.
\newblock {\em Journal of computational and theoretical nanoscience},
  7(9):1759--1770, 2010.

\bibitem{SerreTrees}
J.~Stilwell and J.P. Serre.
\newblock {\em Trees}.
\newblock Springer Monographs in Mathematics. Springer Berlin Heidelberg, 2002.

\bibitem{HermillerShapiro}
Susan Hermiller and Michael Shapiro.
\newblock Rewriting systems and geometric three-manifolds.
\newblock {\em Geometriae Dedicata}, 76(2):211--228, 1999.

\bibitem{GrovesSmith1}
John Richard~James Groves and Geoff~C Smith.
\newblock {\em Rewriting systems and soluble groups}.
\newblock University of Bath, 1989.

\bibitem{GrovesSmith2}
JRJ Groves and GC~Smith.
\newblock Soluble groups with a finite rewriting system.
\newblock {\em Proceedings of the Edinburgh Mathematical Society},
  36(2):283--288, 1993.

\bibitem{Hermiller}
Susan~M. Hermiller.
\newblock Rewriting systems for coxeter groups.
\newblock {\em Journal of Pure and Applied Algebra}, 92(2):137 -- 148, 1994.

\bibitem{BrinkHowlett}
Brigitte Brink and Robert~B Howlett.
\newblock A finiteness property an an automatic structure for coxeter groups.
\newblock {\em Mathematische Annalen}, 296(1):179--190, 1993.

\bibitem{Edelman}
Paul~H Edelman.
\newblock Lexicographically first reduced words.
\newblock {\em Discrete Mathematics}, 147(1-3):95--106, 1995.

\bibitem{CoxeterBook}
Anders Bjorner and Francesco Brenti.
\newblock {\em Combinatorics of Coxeter groups}, volume 231.
\newblock Springer Science \& Business Media, 2006.

\bibitem{Intervening}
H.~{Eriksson} and K.~{Eriksson}.
\newblock {Words with intervening neighbours in infinite Coxeter groups are
  reduced}.
\newblock {\em ArXiv e-prints}, November 2008.

\bibitem{Shende}
Vivek~V Shende, Aditya~K Prasad, Igor~L Markov, and John~P Hayes.
\newblock Reversible logic circuit synthesis.
\newblock In {\em Proceedings of the 2002 IEEE/ACM international conference on
  Computer-aided design}, pages 353--360. ACM, 2002.

\bibitem{Badaoui}
Mohamad Badaoui.
\newblock {\em {G-graphs and Expander graphs}}.
\newblock Theses, {Normandie Universit{\'e}}, March 2018.

\bibitem{Arnaud}
Arnaud (https://cstheory.stackexchange.com/users/4871/arnaud).
\newblock Expander graph from hypergraph.
\newblock Theoretical Computer Science Stack Exchange.
\newblock URL:https://cstheory.stackexchange.com/q/41749 (version: 2018-10-19).

\bibitem{BQPPH}
Scott Aaronson.
\newblock {BQP} and the polynomial hierarchy.
\newblock In {\em Proceedings of the Forty-second ACM Symposium on Theory of
  Computing}, STOC '10, pages 141--150, New York, NY, USA, 2010. ACM.

\bibitem{BirgetOlshanskiiRipsSapir}
J.-C. Birget, A.~Yu Ol'shanskii, E.~Rips, and M.~V. Sapir.
\newblock Isoperimetric functions of groups and computational complexity of the
  word problem.
\newblock {\em Annals of Mathematics}, 156(2):467--518, 2002.

\bibitem{QuantumDynProgDAG}
K.~{Khadiev}.
\newblock {Quantum Dynamic Programming Algorithm for DAGs. Applications for
  AND-OR DAG Evaluation and DAG's Diameter Search}.
\newblock {\em ArXiv e-prints}, April 2018.

\bibitem{GiacomoPark}
G.~{Giacomo Guerreschi} and J.~{Park}.
\newblock {Two-step approach to scheduling quantum circuits}.
\newblock {\em Quantum Science and Technology}, 3(4):045003, October 2018.

\bibitem{VenturelliDoRieffelFrank}
D.~{Venturelli}, M.~{Do}, E.~{Rieffel}, and J.~{Frank}.
\newblock {Compiling quantum circuits to realistic hardware architectures using
  temporal planners}.
\newblock {\em Quantum Science and Technology}, 3(2):025004, April 2018.

\end{thebibliography}

\end{document}